\pdfoutput=1

\RequirePackage{silence}
\documentclass[aps,prx,twocolumn,10pt,nofootinbib]{revtex4-2}

\usepackage[utf8]{inputenc}	
\usepackage[english]{babel}
\usepackage[T1]{fontenc}

\usepackage{amssymb,amsmath,mathtools,amsthm}
\usepackage{dsfont} 
\usepackage{mathrsfs}

\usepackage{enumitem} 
\setlist{nolistsep} 

\usepackage[printonlyused,withpage,nolist]{acronym}
\makeatletter
\AtBeginDocument{%
 \renewcommand*{\AC@hyperlink}[2]{%
   \begingroup
     \hypersetup{hidelinks}%
     \hyperlink{#1}{#2}%
   \endgroup
 }%
}

\usepackage[table,usenames,dvipsnames]{xcolor}
\usepackage[colorlinks=true,
			hyperindex, 
		    linkcolor = NavyBlue,
		    anchorcolor = hyperrefblue,
		    citecolor = Green,
		    filecolor = hyperrefblue,
		    urlcolor = NavyBlue,
			breaklinks=true]{hyperref}

\providecommand{\pra}{Phys.\ Rev.\ A}

\providecommand{\prl}{Phys.\ Rev.\ Lett.}

\newtheorem{corollary}{Corollary}
\newtheorem{proposition}{Proposition}
\newtheorem{definition}{Definition}

\usepackage{braket}
\newtheorem{theorem}{Theorem}
\newtheorem{lemma}{Lemma}
\newtheorem{problem}{Problem}

\newcommand{\e}{\ensuremath\mathrm{e}}
\renewcommand{\i}{\ensuremath\mathrm{i}}

\DeclareMathOperator{\Tr}{Tr}

\DeclareMathOperator{\diag}{diag}

\DeclareMathOperator{\U}{U}

\DeclareMathOperator{\Herm}{Herm} 

\newcommand{\CC}{\mathbb{C}}
\newcommand{\RR}{\mathbb{R}}

\newcommand{\ZZ}{\mathbb{Z}}
\newcommand{\NN}{\mathbb{N}}
\newcommand{\1}{\mathds{1}}

\newcommand{\C}{\mathbb{C}}

\newcommand{\mc}[1]{\mathcal{#1}}

\newcommand{\mcH}{\mc{H}}
\renewcommand{\H}{\mcH}

\newcommand{\argdot}{{\,\cdot\,}}
\renewcommand{\vec}[1]{\boldsymbol{#1}}

\newcommand{\myleft}{\mathopen{}\mathclose\bgroup\left}
\newcommand{\myright}{\aftergroup\egroup\right}

\newcommand{\norm}[1]{\left\Vert #1 \right\Vert} 
\newcommand{\iiiNorm}[1]{{\left\vert\kern-0.25ex\left\vert\kern-0.25ex\left\vert #1 
    \right\vert\kern-0.25ex\right\vert\kern-0.25ex\right\vert}}
\newcommand{\ketbra}[2]{\ket{#1} \!\! \bra{#2}}

\newcommand{\sandwich}[3]
  {\left\langle  #1 \right| #2 \left| #3 \right\rangle}

\newcommand{\func}[1]{{\ensuremath{\mathsf{#1}}}}
\newcommand{\class}[1]{{\ensuremath{\mathsf{#1}}}}
\newcommand{\NP}{\class{NP}}
\newcommand{\poly}{\func{poly}}

\newcommand{\MaxCut}{\class{MaxCut}}

\renewcommand{\P}{\class{P}}
\newcommand{\QMA}{\class{QMA}}
\newcommand{\QCMA}{\class{QCMA}}
\newcommand{\MA}{\class{MA}}
\newcommand{\APX}{\class{APX}}

\newcommand{\modnorm}[1]{\left|#1\right|_{\mathrm{mod}\; 2\pi}}
\newcommand{\OEx}[1]{\left<O(#1)\right>}
\newcommand{\nE}{\left|E(A)\right|}
\newcommand{\apval}{\Delta}
\newcommand{\apvali}{\delta}
\newcommand{\admat}{$A\in \{0,1\}^{d\times d}$}
\newcommand{\MCA}{\MaxCut(A)}
\newcommand{\MCAa}{\mu_a(A)}
\newcommand{\cov}{\Gamma}
\newcommand{\gs}{0}
\newcommand{\obfun}{\mu}
\DeclareMathOperator{\sw}{\Delta \lambda}
\makeatletter 
 \hypersetup{pdftitle = {Training variational quantum algorithms is NP-hard},
	     pdfauthor = {Lennart Bittel, Martin Kliesch},
	     pdfsubject = {Quantum physics},
	     pdfkeywords = {hybrid, quantum, algorithm, 
	     				VQE, variational quantum eigensolver, 
	     				VQA, variational quantum algorithms, 
	     				QAOA, quantum approximate optimization algorithm, 
	     				NP, hard, optimization, minimization, maximization,  
	     				variational, circuit, 
	     				time evolution, }
	    }
\makeatother

\newcommand{\hhu}{%
  Heinrich Heine University D{\"u}sseldorf, 
  Germany%
}

 \definecolor{martin}{rgb}{0,.4,1}

\renewcommand{\sw}{w}
\begin{document}

\title{Training variational quantum algorithms is NP-hard}

\author{Lennart Bittel}
\email{lennart.bittel@uni-duesseldorf.de}
\author{Martin Kliesch}
\email{mail@mkliesch.eu}
\affiliation{\hhu}

\begin{abstract}
\Aclp{VQA} are proposed to solve relevant computational problems on near term quantum devices. Popular versions are \aclp{VQE} and \aclp{QAOA} that solve ground state problems from quantum chemistry and binary optimization problems, respectively. They are based on the idea of using a classical computer to train a parameterized quantum circuit. 

We show that the corresponding classical optimization problems are \NP-hard.  
Moreover, the hardness is robust in the sense that, for every polynomial time algorithm, there are instances for which the relative error resulting from the classical optimization problem can be arbitrarily large 
assuming $\P \neq \NP$. 
Even for classically tractable systems composed of only logarithmically many qubits or free fermions, we show the optimization to be \NP-hard. 
This elucidates that the classical optimization is intrinsically hard and does not merely inherit the hardness from the ground state problem. 

Our analysis shows that the training landscape can have many far from optimal persistent local minima. This means that gradient and higher order descent algorithms will generally converge to far from optimal solutions. 
\end{abstract}

\maketitle


\acresetall
\section{Introduction}
Recent years have seen enormous progress toward large-scale quantum computation. 
A central goal of this effort is the implementation of a type of quantum computation that solves computational problems of practical relevance faster than any classical computer. However, the noisy nature of quantum gates and the high overhead cost of noise reduction and error correction
limit near term devices to shallow circuits~\cite{Pre18}. 

\Acp{VQA} have been proposed to bring us a step closer to this goal. 
Here, an optimization problem is captured by a loss function given by expectation values of observables w.r.t.\ states generated from a parametrized quantum circuit. 
Then a classical computer trains the quantum circuit by optimizing the expectation value over the circuit's parameters. Figure~\ref{fig:sketch} illustrates a possible \ac{VQA} routine.
Popular candidates to be used on near term devices are \acp{QAOA} \cite{FarGolGut14}
and \acp{VQE} \cite{PerMcCSha14}; see Ref.~\cite{Cerezo20VariationalQuantumAlgorithms} for a review.  

\acp{VQE} are proposed, for instance, to solve electronic structure problems, which are central to quantum chemistry and material science. 
Proposals of \acp{QAOA} include improved algorithms for quadratic optimization problems over binary variables such as the problem of finding the maximum cut of a graph (\MaxCut). 
For hybrid classical-quantum computation to be successful, two challenges need to be overcome. 
First, one needs to find parameterized quantum circuits that have the expressive power to yield a sufficiently good approximation to the optimal solution of relevant optimization problems (i.e., the model mismatch is small). 
Second, the classical optimization over the parameters of the quantum circuit needs to be solved quickly enough and with sufficient accuracy. 
We will focus on this second challenge. 

\begin{figure}
	
	\includegraphics[width=1\linewidth]{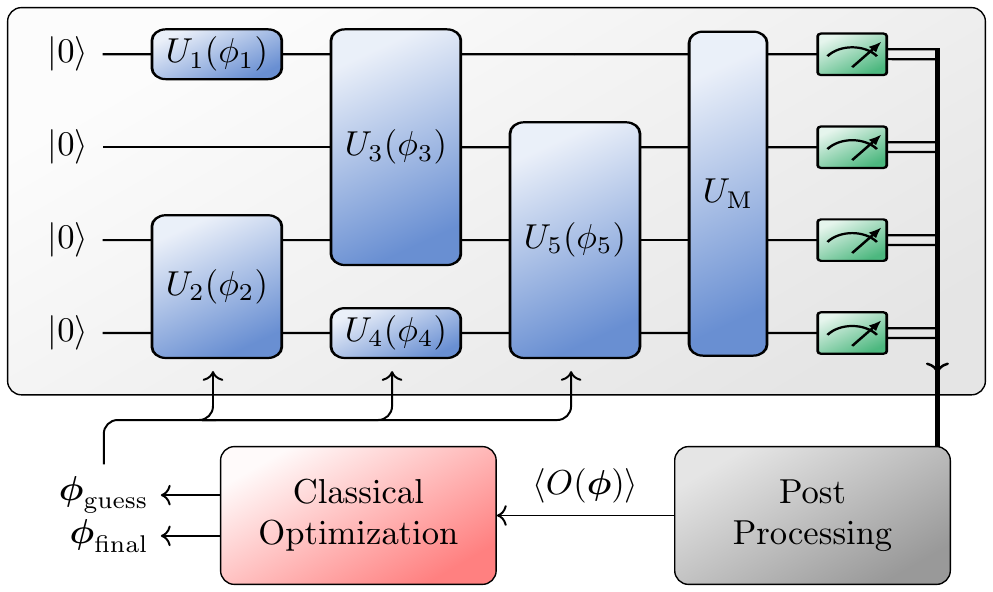}
	\caption{Sketch of a \ac{VQA} optimization routine. 
		This work addresses the complexity of the classical optimization part (red). 
	}
	\label{fig:sketch}
\end{figure}

For the classical optimization several heuristic approaches are known, most of which are based on gradient descent ideas and higher order methods. 
This is convenient, as with the parameter shift rule \cite{Schuld19EvaluatingAnalyticGradients}- the gradient can be calculated efficiently.
Methods include standard BFGS optimization and extensions \cite{Byrd1995LimitedMemoryAlgorithm} and 
natural gradient descent \cite{Stokes2020QuantumNaturalGradient}, which has a favorable performance for at least certain easy instances \cite{Wierichs2020AvoidingLocalMinima}. 
Second order methods require significant overhead in the number of measurements but can yield better accuracy \cite{Mari2020EstimatingTheGradient}. 
Quantum analytic descent \cite{Koczor2020QuantumAnalyticDescent} uses certain classical approximations of the objective function in order to reduce the number of quantum circuit evaluations at the cost of a higher classical computation effort. 

However, it has also been shown recently that there are certain obstacles that need to be overcome to render the classical optimization successful. 
The training landscape can have so-called \emph{barren plateaus} where the loss function is effectively constant and hence yields a vanishing gradient, which prevents efficient training. 
This phenomenon can be caused, for example by random initializations \cite{McClean2018BarrenPlateausIn} and nonlocality of the observable defining the loss function \cite{Cerezo20CostFunctionDependent}. 
Also, sources of randomness given by noise in the gate implementations can cause similar effects \cite{Wang2020Noise-InducedBarren}. 
Moreover, the problem of barren plateaus cannot be fully resolved by higher order methods \cite{Cerezo20ImpactOfBarren}. 

In this work, we show that the existence of persistent local minima can also render the training of \acp{VQA} infeasible. 
For this purpose, we encode the \NP-hard \MaxCut\ problem into the corresponding classical optimization task for several versions of \acp{VQA},
which have many far from optimal local minima. 

Specifically, we obtain hardness results concerning the optimization in four different settings: 
(i) We use an oracle description of a quantum computer and show that the classical optimization of \acs{VQA} is an \NP-hard problem, even if it needs to be solved only within constant relative precision. 
Next, we remove the oracle from the problem formulation by focusing on classically tractable systems where the underlying ground state problem is efficiently solvable. 
Here, we consider quantum systems where (ii) the Hilbert space dimension scales polynomially in the number of parameters (i.e., logarithmically many qubits) or (iii) is composed of free fermions. 
(iv) If the setup is restricted to the \ac{QAOA} type, we show that our hardness results also hold. 

\subsection{Connection to complexity theory}
The decision version of \ac{VQA} optimization is in the complexity class \QCMA, problems that can be verified with a classical proof on a quantum computer. 
The class \QMA, which allows for the proof to be a quantum state, contains \QCMA. 
Much about the relationship between classical \MA, \QCMA, and \QMA\ is still unknown. 
Notably, finding the ground state energy of a local Hamiltonian is \QMA-hard \cite{KitSheVya02,Kempe04TheComplexityOf}. 
This means that if $\QCMA\neq \QMA$, then \ac{VQA} algorithms will not be able to solve the local Hamiltonian problem, but only problems contained in \QCMA. 
Our results imply that even if the relevant energy eigenstates are contained in the \ac{VQA} ansatz, the classical optimization may still be at least as difficult as solving \NP\ problems (Section~\ref{sec:VQAopt_q}). 

\subsection{Notation}
We use the notation $[n]\coloneqq \{1,\dots,n\}$. 
The Pauli matrices are denoted by $\sigma_x$, $\sigma_y$, and $\sigma_z$. 
An operator $X$ acting on subsystem $j$ of a larger quantum system is denoted by $X^{(j)}$ -
e.g., $\sigma_x^{(1)}$ is the Pauli-$x$ matrix acting on subsystem $1$. 
By $\|X\|$ we refer to the operator norm of operator $X$. 

The number of edges of the graph with the adjacency matrix $A$ is denoted by $\nE$. 
By $\MCA$ we denote the solution of \MaxCut\ for an adjacency matrix $A$; see Problem~\ref{p:maxcut}. 

Throughout, we consider only adjacency matrices $A$ of undirected, unweighted graphs with at least one edge; i.e., \admat\ is a nonzero symmetric binary matrix with a vanishing diagonal. 


\section{A continuous \texorpdfstring{\NoCaseChange{\MaxCut}}{MaxCut} optimization}
We introduce a continuous, trigonometric problem which we show to be \NP-hard to optimize and approximate. 
This is related to earlier work on the optimization of trigonometric functions~\cite{Pfister2018BoundingMultivariateTrigonometric} for which \NP-hardness is known. For the specific class of functions, we also show the existence of an approximation ratio explicitly. 
Below, we use this problem to obtain hardness results for various \ac{VQA} versions. 

	\begin{problem}[\MaxCut]\label{p:maxcut}\hfill\vspace{-.8\baselineskip} \\
		\begin{description}[noitemsep,leftmargin=0.5cm,font=\normalfont]
			\item [Instance] 
			The adjacency matrix \admat\ of an unweighted undirected graph. 

			\item [Task]
			Find $S\subset [n]$ that maximizes $\sum_{i\in S,j\in [n]\setminus S} A_{i,j}$. 
		\end{description}
	\end{problem}
	{\MaxCut} is famously known to be \NP-hard.
	Additionally \MaxCut\ is \APX -hard, meaning that every polynomial time approximation algorithm
	there exist some instances, where the approximation ratio $\alpha$, the ratio between the algorithmic solution and the optimal solution, is bounded by $\alpha\leq \alpha_{\max}<1$, assuming that $\P\neq \NP$. 
	It was shown that if the unique games conjecture is true, then the best approximation ratio of a polynomial algorithm is $\alpha_{\max}=\min_{0<\theta<\pi}\frac{\theta/\pi}{(1-\cos(\theta)/2)}\approx 0.8786$~\cite{khot_optimal_2007}, which is also what the best known algorithms can guarantee~\cite{goemans_improved_1995}. Without the use this conjecture, it has been proven that $\alpha_{\max}\leq\frac{16}{17}\approx0.941$~\cite{hastad_optimal_2001}.
For our purposes we define a continuous, trigonometric version of \MaxCut. 
Minima of real valued functions are given by real numbers that may not have an efficient numerical representation. 
However, it is commonly said that a minimization problem is solved if it is solved to exponential precision, which is the convention we will also be using throughout this paper. The intuitive notion is that the hardness does not come from the difficulty of representing the minimum.
	\begin{problem}[Continuous \MaxCut]\label{p:shift_maxcut}\hfill\vspace{-.8\baselineskip} \\
		\begin{description}[noitemsep,leftmargin=0.5cm,font=\normalfont]
			\item [Instance] 
			The adjacency matrix \admat\ of an unweighted graph.
			
			\item [Task] 
			Find $\vec\phi\in[0,2\pi)^d$ that minimizes 
			\begin{equation}
				\obfun (\vec \phi)\coloneqq \frac{1}{4}\sum_{i,j=1}^dA_{i,j}[\cos(\phi_i)\cos(\phi_j)-1].
			\end{equation}
		\end{description}
	\end{problem}
\begin{lemma}\label{lem:continuous_MaxCut}
	Problem~\ref{p:shift_maxcut} is \NP-hard. 
	Moreover, if $\P\neq\NP$, 
	for every polynomial time algorithm there exists an approximation ratio, which is at most 
	that
	of \MaxCut. 
\end{lemma}
\begin{proof}
 We will show that it suffices to look at $\vec\phi$ from the discrete subset $\{0,\pi\}^d$.  
 For this purpose, we analyze the dependence of $\obfun$ on one coordinate $\phi_i$ of $\vec \phi$. 
 Denoting the vector obtained from $\vec\phi$ by replacing $\phi_i$ with $x$
 by $\left.\vec\phi\right|_{\phi_i=x}$, 
 we write $\obfun$ as 
 \begin{align}
 \obfun(\left. \vec \phi\right|_{\phi_i=x})
 =
 \cos(x) \left[\frac{1}{2}\sum_{j:\; j\neq i} A_{i,j}\cos(\phi_j)\right]+C\,,
 \end{align}
 where we have used that $A_{i,i} = 0$ for all $i$, $A=A^T$ and $C$ is independent of $x$. 
 Since the only dependence on $x$ is given by the cosine, 
 it follows that for any $\vec \phi\in [0,2\pi)^d$
 \begin{equation}
 \obfun(\vec\phi)\geq\min\{\obfun(\left.\vec\phi\right|_{\phi_i=0}),\obfun(\left.\vec\phi\right|_{\phi_i=\pi})\}\,. 
 \end{equation} 	
 This observation implies that we can iteratively choose each $\phi_i$ to be in $\{0,\pi\}$ for $i\in[d]$, without increasing $\obfun$. Therefore, an algorithm returning continuous values of $\phi_i$ can be turned discrete in polynomial time without reducing the approximating power.
 The discrete problem can be written as a quadratic unconstrained binary optimization
 \begin{equation}
 \begin{aligned}
 \min_{\vec\phi\in \{0,\pi\}^d}\obfun(\vec \phi)
 &=
 \frac{1}{4}\min_{\vec v \in\{-1,1\}^d} \sum_{i, j=1}^dA_{i,j}(v_iv_j-1)\\
 &=\frac{1}{4}
 \min_{\vec v \in\{-1,1\}^d} \sum_{i,j=1}^d A_{i,j}\left(-2\delta_{v_i\neq v_j}\right)\\
 &=-\max_{S \subset [d]} \sum_{i\in S,j\in [d]\setminus S} A_{i,j} \, ,
 \end{aligned}
 \end{equation}
 where $\delta_{v_i\neq v_j}=1-\delta_{v_i,v_j}$ and we used again that $A$ is symmetric. 
 Therefore, we obtain a  (many-one) reduction of {\MaxCut} to Problem~\ref{p:shift_maxcut}, implying Problem~\ref{p:shift_maxcut} is \NP-hard, which finishes the proof. 
\end{proof}

We note that the derivative and the Hessian of $\mu$ is given by
\begin{equation}
\begin{aligned}
\frac{\partial\obfun(\vec\phi)}{\partial \phi_i}
&=-\frac 12 \sin(\phi_i)\sum_{j:\; j\neq i} A_{i,j}\cos(\phi_j)\\
\frac{\partial^2\obfun(\vec\phi)}{\partial \phi_i\partial \phi_k}
&=-\frac 12\delta_{i,k} \cos(\phi_i)\sum_{j:\; j\neq i} A_{i,j}\cos(\phi_j)\\
&+\frac{A_{i,j}}{2} \sin(\phi_i)\sin(\phi_k) \, .
\end{aligned}
\end{equation}
At the relevant discrete set $\vec\phi\in\{0,\pi\}^d$, the derivative vanishes and the Hessian is diagonal,
meaning a point in the set describes a local minimum whenever changing any single $\phi_i$ increases the objective function. 
The same minima (in the discrete \MaxCut\ formulation) are also achieved by a greedy algorithm: 
start with a random bipartition of the vertex set, then repeatably change a single vertex assignment if it increases the cut until the cut cannot be increased any further by this update rule. 
This algorithm has an approximation ratio of $\alpha=\frac{1}{2}$, meaning it only guarantees to approximate \MaxCut\ to half its optimal solution.
If Problem~\ref{p:shift_maxcut} is solved with gradient based methods, any local minimum can be the final result, therefore gradient based algorithms also only have an approximation ratio of $\alpha=\frac 1 2$, which is significantly worse than what modern \MaxCut\ solvers can achieve~\cite{goemans_improved_1995}.

\section{\texorpdfstring{\Aclp{VQA}}{VQAs}}
\Ac{VQA} is a general framework of hybrid quantum computers, where classically tunable parameters $\vec \phi$ of a unitary circuit are used to minimize the expectation value of an observable. 
First, in Section~\ref{sec:VQAopt_q}, we consider such computing schemes, where the quantum part is composed of qubits. 
In order to show that the classical simulation is hard, we assume oracle access to an idealized quantum device. 
Next, in Section~\ref{sec:VQEopt_poly}, we show that the problem is also hard for \ac{VQA} settings with small Hilbert space dimensions (or logarithmically many qubits) so that the oracle can be replaced by efficient classical simulation. 
In Section~\ref{sec:VQAopt_QAOA} we use the same setting, but consider \acp{QAOA} instances instead.
Last, in Section~\ref{sec:VQAopt_ff}, we analyze \acp{VQA} in free Fermionic systems, where the oracle can be replaced by efficient free fermionic calculations.

\subsection{\texorpdfstring{\Ac{VQA}}{VQA} optimization with quantum computer access} 
\label{sec:VQAopt_q}
The common application of \acp{VQA} is within quantum computing, where a quantum computer is used to estimate the expectation value and a classical algorithm chooses the circuit parameters of the quantum computer. 
For the classical optimization, we describe the information obtained from the quantum computer with oracle calls made by the classical algorithm. 

\begin{problem}[\Ac{VQA} minimization, oracular formulation]\label{p:VQA_O}\hfill
	\vspace{-.8\baselineskip} 
	\begin{description}[noitemsep,leftmargin=0.5cm,font=\normalfont]
		
		\item[Instance] 
		A set of generators $\{H_i \}_{i \in \{1,\dots, L\}}$ 
		and an observable $O$ acting on $\H=(\CC^2)^{\otimes N}$, 
		given in terms of their Pauli basis representation. 
		
		\item[Oracle access] 
		We set $\ket{\Psi(\vec \phi)}\coloneqq U_L(\phi_L)\cdots U_1(\phi_1)\ket{\vec 0}$ with $U_i(\phi)=\e^{-\i H_i \phi}$. 
		 The oracle $\mc O$ returns
		$\OEx{\vec \phi} \coloneqq \sandwich{\Psi(\vec \phi)}{O}{\Psi(\vec \phi)}$, for a given $\vec \phi$, up to any desired polynomial additive error. 
		
		\item[Task]
		Find $\vec \phi \in \RR^L$ that minimizes $\OEx{\vec \phi}$ provided access to $\mc O$. 
	\end{description}
\end{problem}
We use the oracle to outsource difficult computations, which is similar to how a quantum computer would in a physical implementation. The motivation of our oracle is that Problem~\ref{p:VQA_O} captures the complexity of only the classical optimization effort in hybrid quantum computations. 
The oracle can be seen as postselecting on the successful runs only, therefore making the return deterministic. 

\begin{proposition}[Hardness of \ac{VQA} optimization, oracular formulation]
	\label{thm:VQA1}
	Assuming $\P\neq \NP$ there is no deterministic classical algorithm that solves Problem~\ref{p:VQA_O} in polynomial time. 
\end{proposition}

It is straightforward to show that Problem~\ref{p:VQA_O} is \NP-hard to solve.
Essentially, 
we use a diagonal observable for which the ground state problem is \NP-hard and use unitaries to reach every computational basis state.
\begin{proof}
	We prove the proposition via a reduction of Problem~\ref{p:shift_maxcut} to Problem~\ref{p:VQA_O}. For this, let $N=d$ and let $O$ the usual Ising Hamiltonian encoding of \MaxCut
	\begin{align}
	O&\coloneqq\frac{1}{4}\sum_{i,j=1}^dA_{i,j}(\sigma_z^{(i)}\sigma^{(j)}_z-1)\, .
	\end{align}
	We use $L=d$ layers with
	\begin{align}
		H_i&\coloneqq \frac{\sigma_y^{(i)}}{2} \quad,\quad i\in [d] \,,
	\end{align}
	as generators. 
	By direct calculation we find that 
	\begin{equation}
	\begin{aligned}
	\OEx{\vec \phi}&=\bra{\Psi(\vec\phi)}O\ket{\Psi(\vec \phi)}\\
	\label{eq:oracle_replacement}
	&=\frac{1}{4}\sum_{i,j=1}^dA_{i,j}\left[\cos(\phi_i)\cos(\phi_j)-1\right]\\
	&=\obfun(\vec\phi)
	\, ,
	\end{aligned}
	\end{equation}
	which is the objective function of Problem~\ref{p:shift_maxcut}.  
\end{proof}

To analyze the overall approximation power of an algorithm we define the \emph{approximation error} for an instance as 
\begin{align}
\apvali\coloneqq \frac{\braket{O}_a-\lambda_{\min}(O)}{\lambda_{\max}(O)-\lambda_{\min}(O)}\, ,
\end{align}
where $\lambda_{\min}(O)$ is the smallest eigenvalue of the observable $O$ and $\lambda_{\max}(O)$ is the largest; the expectation value of the final output of the algorithm is $\braket{O}_a\geq \lambda_{\min}(O)$. 
We normalize by the \emph{spectral width}
\begin{align}
\sw(O)\coloneqq\lambda_{\max}(O)-\lambda_{\min}(O)\,,
\end{align} as this ensures that $\apvali\in [0,1]$.
There are two error contributions: 
(i) the \emph{model mismatch} $\apvali_m$ is the approximation error resulting from the ansatz class being unable to represent the ground state and 
(ii) the \emph{optimization error} $\apvali_o$ is the error due to the classical algorithm not converging to the optimal solution within the class. 
That is, 
\begin{align}
\apvali&=
\frac{\braket{O}_{\min}-\lambda_{\min}(O)}{\sw(O)}+\frac{\braket{O}_a-\braket{O}_{\min}}{\sw(O)}\\
&=\qquad\apvali_{\mathrm{m}}\qquad+\qquad\apvali_o 
\, , 
\end{align}
where $\braket{O}_{\min}$ refers to the smallest expectation value over the ansatz class, i.e., the global minimum over the circuit parameters. 
Since we are interested in classical algorithms, we define an optimization error, in a similar manner to how approximation ratios are defined for \NP\ optimization problems (the complexity class \APX), over all considered instances.
\begin{definition}[Optimization error]\label{def:approx_val}
	The \emph{optimization error} of an optimization algorithm $\apval\in [0,1]$ is the smallest number such that
	\begin{align}
	\apval \geq \frac{\braket{O}_a-\braket{O}_{\min}}{\sw(O)}
	\end{align}
	for all considered \ac{VQA} instances. 
\end{definition}

\begin{corollary}\label{cor:apvalexp}
	If $\P\neq\NP$, 
	then there exists no polynomial time algorithm which can guarantee any optimization error $\apval<1$ for all \acp{VQA} defined by Problem~\ref{p:VQA_O}.
\end{corollary}

\begin{proof}
	We prove this statement by relating the optimization error of a \ac{VQA} to the approximation ratio of \MaxCut\ and by introducing a boosting technique to amplify errors in the setting of Problem~\ref{p:VQA_O}.
	
	From the proof of Proposition~\ref{thm:VQA1} we obtain $\sw(O)=\MCA$ and the optimal solution is also $\left|\braket{O}_{\min} \right|=\MCA$, as there is no model mismatch. 
	From the algorithm we get $\braket{O}_{a}=\MCAa$, where $\MCAa$ is the approximation of the continuous \MaxCut\ problem (Problem~\ref{p:shift_maxcut}) and, therefore, an approximation to \MaxCut\ itself. 
	With this argument it follows that 
	\begin{equation}
		\apval\geq1-\alpha\, ,
	\end{equation}
	where $\alpha$ is the approximation ratio related to Problem~\ref{p:shift_maxcut} of the algorithm. 
	To boost this result we introduce a variable $k$ and choose operators for $k\times d$ qubits
	\begin{align}
		\tilde O&=(-1)^{k-1}O^{\otimes k} \, ,\\
		\tilde U(\vec \phi)&=U(\vec \phi)^{\otimes k} \, .
	\end{align}
	We can verify that the generators and $\tilde O$ only have $\poly(d)$ many terms for constant $k$. For the expectation value this gives
\begin{equation}
	\begin{aligned}
		\braket{\tilde O(\vec \phi)}&=(-1)^{k-1}\bra{\Psi(\vec \phi)}^{\otimes k}O^{\otimes k}\ket{\Psi(\vec \phi)}^{\otimes k}\\
		&=-\left|\OEx{\vec\phi}\right|^k=-|\obfun(\vec\phi)|^k \, ,
	\end{aligned}
\end{equation}
	where the introduced sign ensures that the problem remains a minimization for all $k$.
	We obtain $\sw(\tilde O)=\bigl|\braket{\tilde O}_{\min}\bigr|=\MCA^k$. 
	This yields
\begin{equation}
	\begin{aligned}
	\apval&\geq \sup_A \frac{\left|\MCA^k-\left|\braket{\tilde O}_a\right|\right|}{\MCA^k}\\
	&=\sup_A \left(1-\frac{|\MCAa|^k}{\MCA^k}\right)\\
	&= (1-\alpha^k)\,, 
	\end{aligned}
\end{equation}
	Therefore, no optimization error strictly smaller than $\apval<1$ can exist for all instances (if $\P\neq \NP$), as this would mean in return, the algorithm could solve Problem~\ref{p:shift_maxcut} to arbitrary precision.
\end{proof}

\subsection{Logarithmic number of qubits --- polynomial Hilbert space dimension}
\label{sec:VQEopt_poly}
We can improve on the previous result by allowing only $N\in O(\log(d))$ many qubits, where $d$ is the input length of the \MaxCut\ instance. 
This drastically reduces the system's size and complexity. Notably, since the Hilbert space is now only of polynomial dimension, both the calculation of expectation values and the ground state problem can be computed efficiently. 
Yet we show that \ac{VQA} optimization is still \NP-hard.
This means that the classical optimization does not merely inherit the hardness of the ground state problem but rather is intrinsically difficult.
Since the expectation value is efficiently numerically simulatable, we do not require oracle access to a quantum computer to analyze the problem. 
Also, for convenience, instead of the Pauli-basis we use the computational basis of the Hilbert space $\H$ of dimension $\dim(\H)=2^N\eqqcolon n$.
This gives the following problem description. 

\begin{problem}[\Ac{VQA} minimization problem]\label{p:VQA}\hfill
\begin{description}[noitemsep,leftmargin=0.5cm,font=\normalfont]
\item [Instance] 
An initial state $\ket{\Psi_0}\in \CC^n$, a set of generators $\{H_i \}_{i \in \{1,\dots, L\}}\subset \Herm(\CC^n)$, where $L$ is the number of layers and an observable $O\subset \Herm(\CC^n)$. 

\item[Task]
For $\ket{\Psi(\vec \phi)}\coloneqq U_L(\phi_L)\cdots U_1(\phi_1)\ket{\Psi_0}$ with $U_i(\phi)=\e^{-\i H_i \phi}$, 
find a $\vec \phi \in \RR^L$ that minimizes $\OEx{\vec\phi}\coloneqq~\sandwich{\Psi(\vec \phi)}{O}{\Psi(\vec \phi)}$. 
		\end{description}
\end{problem}

\begin{theorem}[Hardness of \ac{VQA} optimization]\label{thm:VQA2}
	\Ac{VQA} optimization (Problem~\ref{p:VQA}) is \NP-hard. 
\end{theorem}

\begin{proof} 
We prove the theorem via a many-one reduction from Problem~\ref{p:shift_maxcut}
Let \admat\ be the adjacency matrix of an unweighted graph. 
On the Hilbert space $\H=\CC^{2d}$
we first define an observable in the standard basis as
\begin{align}\label{eq:defOprime}
	O'&\coloneqq \frac{d}{8}\cdot A\otimes \begin{pmatrix} 1&1\\1&1	\end{pmatrix}\,,
\end{align}
where $\otimes$ denotes the Kronecker product. For the actual observable we modify the diagonal as 
\begin{align}\label{eq:Observable_loc_VQA}
	O_{i,j}=\left\{\begin{matrix}
	O'_{i,j}& i\neq j\\
	-\sum_{\alpha=1}^{2d} O'_{\alpha,j} & i=j\\
	\end{matrix}\right.\,.
\end{align}  
The initial state and generators are chosen as
	\begin{align}
	\ket{\Psi_0}&\coloneqq \frac{1}{\sqrt{2d}} \sum_{j=1}^{2d} \ket{j} \, ,\\\label{eq:VQA_U_small}
	H_i&\coloneqq \ketbra{2i-1}{2i-1}- \ketbra{2i}{2i}\,,
	\end{align} 
	where we take $L=d$ layers.
	As the parametrized state we obtain
\begin{equation}
	\begin{aligned}
	\ket{\Psi(\vec\phi)}
	&\coloneqq 
	U_{d}(\phi_d)\dots U_1(\phi_1) \ket{\Psi_0} 
	\\
	&=
	\frac{1}{\sqrt{2d}}\sum_{j=1}^{2d}\left( \e^{-\i \phi_j}\ket{2j-1}+\e^{\i \phi_j}\ket{2j} \right)
\end{aligned}
\end{equation}
and 
\begin{equation}\label{eq:VQA}
	\begin{aligned}
	\OEx{\vec \phi}&=\bra{\Psi(\vec \phi)}O\ket{\Psi(\vec \phi)}
	\\
	&=\frac{1}{16} \sum_{s,p\in\{+,-\}}\sum_{i,j=1}^d\e^{\i s \phi_i }A_{i,j} \e^{-\i p \phi_j }-\frac{1}{4}\sum_{i,j=1}^d A_{i,j} \\
	&=\frac{1}{8}\sum_{i,j=1}^d A_{i,j} \left(\cos(\phi_i-\phi_j)+\cos(\phi_i+\phi_j)-2\right)\\
	&=\frac{1}{4}\sum_{i,j=1}^d A_{i,j}\left[\cos(\phi_i)\cos(\phi_j)-1\right]\\
	&=\obfun(\vec\phi)
	\end{aligned}
\end{equation}
as corresponding expectation value. 
	This completes the reduction of Problem~\ref{p:shift_maxcut} to Problem~\ref{p:VQA}.
\end{proof}
From this result, \NP-completeness follows for the decision version.
\begin{problem}[\Ac{VQA} minimization, decision version]\label{p:VQA_d}\hfill
	\begin{description}[noitemsep,leftmargin=0.5cm,font=\normalfont]
	\item[Instance] An initial state $\ket{\Psi_0}\in \CC^n$, a set of generators $\{H_i \}_{i \in \{1,\dots, L\}}\subset \Herm(\CC^n)$, where $L$ are the number of layers, an observable $O\in\Herm(\CC^n)$ and a threshold $a\in \RR$. 
	
	\item[Task] 
	For $\ket{\Psi_{\vec \phi}}\coloneqq U_L(\phi_L)\cdots U_1(\phi_1)\ket{\Psi_0}$ with $U_i(\phi)~=~\e^{-\i H_i \phi}$, 
	determine wheter there exists 
	$\vec \phi \in \RR^d$ for which $\sandwich{\Psi(\vec \phi)}{O}{\Psi(\vec \phi)}\leq a$. 
\end{description}
\end{problem}

\begin{corollary}
	Problem~\ref{p:VQA_d} is \NP-complete.
\end{corollary}
\begin{proof}
	As calculating the expectation value of observable on polynomial dimensional Hilbert spaces is in \P, $\vec \phi$ is a valid proof for the \emph{yes} instances, which can be verified in polynomial time and is therefore in \NP. Together with hardness of problem~\ref{p:VQA}, this means Problem~\ref{p:VQA_d} is \NP-complete.
\end{proof}
We now show that $L=1$ layer is sufficient to show hardness. 
For this purpose we will use certain properties of Hamiltonian spectra. 

\begin{definition}[Approximate ergodic energy spectrum]
Let $\epsilon>0$. 
We call a set $\{E_i\}_{i\in n}\subset \RR$ an \emph{$\epsilon$-approximate ergodic energy spectrum} if for all 
$\vec \phi \in [0,2\pi)^n$ there exists $t\in \RR_0^+$
such that 
\begin{equation}
	\modnorm{\phi_i - E_i t}
	\leq
	\epsilon
\end{equation}
for all $i\in [n]$, 
where 
$\modnorm{x} \coloneqq
\inf_{k\in \ZZ} |x-2\pi k| \in [0,\pi]
$. 
\end{definition}

Generic energy spectra are exactly ($\epsilon=0$) ergodic. 
For our purpose we want to show that there are also efficiently expressible approximate ergodic energy spectra.  

\begin{lemma}[Approximate ergodic energy spectra]\label{l:ph_to_en}\hfill
\\
Let $m\in \NN$. 
Then
\begin{equation}
E_i\coloneqq \frac{2\pi}{m^i}	
\end{equation}
with $i \in [n]$ defines an $\epsilon$-approximate ergodic energy spectrum with 
\begin{equation}
	\epsilon = \frac {4\pi} m\, .
\end{equation}

\end{lemma}

We provide a proof in Appendix~\ref{ap:erd_en}. The chosen energies can be expressed with $n\times \lceil\log_2(m)\rceil$ bits of precision. With this property we can show the following theorem.

\begin{theorem}\label{t:unitary}
	\Ac{VQA} optimization (Problem~\ref{p:VQA}) is \NP-hard for $L=1$ layer.
\end{theorem}

	\begin{proof}
	For the single layer, we choose the generators as a linear combination  of terms from Eq.~\eqref{eq:VQA_U_small}
	\begin{align}\label{eq:Ham_ev}
	H=\sum_{j=1}^d E_j \left( \ket{2j-1}\bra{2j-1}- \ket{2j}\bra{2j}\right)\,
	\end{align}
	meaning $U(\phi)=\exp(-\i \phi H)=U(\vec E \phi)$. The initial state and $O$ remain identical. This leads to the expectation value
	\begin{align}
		\OEx{\phi}=\sum_{i,j=1}^d A_{i,j}\left[\cos(E_i \phi)\cos(E_j \phi)-1\right]=\obfun(\vec E\phi)\,.
	\end{align}
If $\{E_i\}_{i}$ are chosen as in Lemma~\ref{l:ph_to_en} then $\OEx{\phi}$ approximates $\obfun(\vec\phi)$ with $\vec \phi = \vec E \phi$ to arbitrary precision, 
which we have shown to be \NP-hard to optimize in Lemma~\ref{lem:continuous_MaxCut}. 
\end{proof}

By viewing the \ac{VQA} in Theorem~\ref{t:unitary} as a continuous time evolution for logarithmically many qubits, we obtain the following result (we are unaware of this statement being explicitly proven before). 

\begin{corollary}
For a system with logarithmically many qubits, we consider the expectation value of a (unitarily) time evolved observable $\langle O(t)\rangle$, starting from some initial state. 
Minimizing the expectation value over $t\in \RR_0^+$ is then \NP-hard.
\end{corollary}

\subsection{\texorpdfstring{\Aclp{QAOA}}{QAOA} for a logarithmic number of qubits}
\label{sec:VQAopt_QAOA}
\Acp{QAOA} can be seen as certain types of \acp{VQA}, which are inspired by adiabatic computation, where a slow enough transition between two Hamiltonians $H_b$ and $H_c$ guarantees remaining in the ground state as long as the Hamiltonians are gapped and level crossings are avoided~\cite{Albash2018AdiabaticQuantumComputation}. 
\Acp{QAOA} capture a time-discrete version of this approach by alternatingly applying the time evolutions of the Hamiltonians.
Accordingly, parameter vectors $\vec \beta,\vec \gamma \in \RR^{L}$ need to be chosen, which define how long each Hamiltonian is applied. 

We demonstrate that the hardness of \ac{VQA} optimization for logarithmically many qubits also translates to \ac{QAOA} problems.  

Formally, the problem is as follows. 
\begin{problem}[\Ac{QAOA} minimization problem]\label{p:QAOA}\hfill
	\begin{description}[noitemsep,leftmargin=0.5cm,font=\normalfont]
		\item[Instance] Two Hamiltonians $H_b,H_c\in \Herm(\CC^n)$ and the number of layers $L$ in unary notation\footnote{This means that the length of the input scales linearly with $L$.}.
		
		\item[Task] For a tunable state $\ket{\Psi(\vec \beta, \vec \gamma)}\coloneqq U_b(\beta_L)U_c(\gamma_L)\cdots U_b(\beta_1)U_c(\gamma_1) \ket{\Psi_0}$, where $\ket{\Psi_0}$ is the ground state of $H_b$, $U_b(\beta)=\e^{-\i H_b \beta}$ and $U_c(\gamma)=\e^{-\i H_c \gamma}$,
		find $\vec \beta,\vec \gamma \in \RR^d$ which minimize $\OEx{\vec \beta, \vec \gamma}\coloneqq\sandwich{\Psi(\vec \beta, \vec \gamma)}{H_c}{\Psi(\vec \beta, \vec \gamma)}$. 
	\end{description}
\end{problem}

\begin{theorem}[Hardness of optimization in \acp{QAOA}]\label{thm:QAOA}
	Problem~\ref{p:QAOA} is \NP-hard for $L=1$ layer.
\end{theorem}
\begin{proof}
	We will perform a reduction from single layer \ac{VQA} to \ac{QAOA}, which implies that Problem~\ref{p:QAOA} is \NP-hard.
	
	We consider the Hilbert space $\H = \CC^{2d+1}$. 
	For $H_b$ we take
	\begin{align}
	H_b&=\mathrm{diag}(E_1,-E_1,E_2,-E_2,\dots ,E_d,-E_d,-1)\, ,
	\end{align}
	where $|E_i|<1$ for all $i\in[d]$.  For $H_c$
	\begin{align}
	H_c&=O\oplus 0 +\tau\left(\ket{+_{2d}}\bra{2d+1}+\ket{2d+1}\bra{+_{2d}}\right)\, ,
	\end{align}
	where $\ket{+_{2d}}=\sum_{j=1}^{2d}\ket{j}/\sqrt{2d}$, $\tau$ is some real constant that we adjust later and the observable $O$ is as defined in Eq.~\eqref{eq:Observable_loc_VQA}. $O\oplus 0$ refers to $O$ being embedded in the first $2d$ computational states in the Hilbert space.
	By design, $\lambda_{\min}(H_b)=-1$ is the ground state energy with ground state $\ket{2d+1}$. 
	\\
	For the state we obtain
\begin{equation}
	\begin{aligned}
	\ket{\Psi(\gamma)}&\coloneqq U_c(\gamma)\ket{2d+1}\\
	&=\cos(\tau\gamma)\ket{2d+1}+\i\sin(\tau\gamma)\ket{+_{2d}}\\
		\end{aligned}
	\end{equation}
	after applying the first Hamiltonian, where we used that $\ket{+_{2d}}$ is an eigenstate of $O\oplus 0$. This gives the final state 
	\begin{equation*}
	\begin{aligned}
	\ket{\Psi(\beta,\gamma)}&=U_b(\beta)U_c(\gamma)\ket{2d+1}
	\\
	&=
	\cos(\tau\gamma)\e^{\i\beta}\ket{2d+1}\\&+\i\sin(\tau\gamma)\frac{1}{\sqrt{2d}}\sum_{j=1}^d\e^{-\i E_j\beta}\ket{2j-1}+\e^{\i E_j\beta}\ket{2j}\, .
	\end{aligned}
\end{equation*}
	From this variational state we derive the expectation value
	\begin{equation}\nonumber
		\begin{aligned}
	\OEx{\beta,\gamma}&=\bra{\Psi(\beta,\gamma)} H_c\ket{\Psi(\beta,\gamma)}\\
	&=\sin^2(\tau\gamma)f(\beta)+2\tau\cos(\tau\gamma)\sin(\tau\gamma)g(\beta)
	\end{aligned}
	\end{equation}
	with
	\begin{align}\label{eq:O_QAOA1}
	f(\beta)&=\frac{1}{4}\sum_{i,j}(A_{i,j}\left[\cos(E_i\beta)\cos(E_j\beta)-1\right]\, ,\\
	g(\beta)&=-\frac{\sin(\beta)}{d}\sum_{i=1}^d\cos(E_j\beta)\,.
	\end{align}
	For $\tau\ll1$ the contribution of $g$ becomes insignificant and $\gamma=\frac{\pi}{2\tau}$ minimizes the objective function as $f(\beta)\leq0$, meaning the problem is equivalent to minimizing $f(\beta)=\obfun(\vec E \beta )$ which approximates $\obfun(\vec\phi)$ to arbitrary precision  if $\vec E$ is chosen as in Lemma~\ref{l:ph_to_en} and therefore gives a reduction from Problem~\ref{p:shift_maxcut}.
\end{proof}

In the proof of Theorem~\ref{thm:QAOA}, the energies of $H_b$ span many orders of magnitude. This means a potential quantum computer needs to be incredibly precise to implement such a \ac{QAOA}.
We will show, that this is not required, but that also for very simple spectra, hardness results can be obtained.
\begin{theorem}\label{thm:QAOA2}
	QAOA optimization (Problem~\ref{p:QAOA}) is \NP-hard for periodic optimization $\vec \beta,\vec\gamma\in [0,2\pi)^L$, even if we restrict $\|H_b\|\leq3$ and $\|H_c\|\leq3$ (i.e $E_i\in \{-3,-2,\dots,3\}$).
\end{theorem}

\begin{proof}[Proof outline]
	In Appendix~\ref{ap:multiLQAOA}, we construct explicit Hamiltonians $H_b$ and $H_c$ from an adjacency matrix $A$, where the solution is $\left<H_c\right>_{\min}=1-\frac{2\MCA}{\nE}$ and $\|H_c\|=1$. We do this by embedding a modified version of Problem~\ref{p:shift_maxcut} into the \ac{QAOA} circuit in such a way, that deviations from the intended structure are penalized by increasing the expectation value. 
\end{proof}
From this we can derive bounds on the optimization errors for \ac{VQA} in this restricted setting. Here, we are unable to use the same boosting technique to increase the hardness result further.
\begin{corollary}
	All polynomial time algorithms for \Ac{QAOA} and therefore \ac{VQA} optimization (Problems~\ref{p:QAOA} and \ref{p:VQA}) have an optimization error $\apval\geq \frac{1-\alpha_{\max}}{2}$, where $\alpha_{\max}$ is the approximation ratio of \MaxCut.
\end{corollary}
\begin{proof}
	For the Hamiltonians in proof of Theorem~\ref{thm:QAOA2}, we have $\sw(H_c)=2$ and the lowest achievable expectation value is $\left<H_c\right>_{\min}=1-\frac{2\MCA}{\nE}$, where $\nE$ is the number of edges of the the graph. From this we can calculate an upper limit on the possible guaranteed precision of an optimization algorithm for all instances
	\begin{equation}
	\begin{aligned}
	\apval\geq&\sup_A\left(\frac{\left|\left<H_c\right>_{\min}-\left<H_c\right>_{\mathrm{a}}\right|}{\sw(H_c)}\right)\\
	&\geq\frac{1}{2}\sup_A\left|1-\frac{2\MCA}{\nE}-\left(1-\frac{2|\MCAa|}{\nE}\right)\right|\\
	&=\frac{1}{2}\sup_A\left(\frac{2\MCA}{\nE}-\frac{2|\MCAa|}{\nE}\right)\\
	&=\frac{1}{2}(1-\alpha)\frac{2\MCA}{\nE}\\
	&\geq \frac{1-\alpha}{2}\, ,
	\end{aligned}
	\end{equation}
	where  the supremum goes over all adjacency matrices and $|\MCAa|$ is the approximation of \MaxCut\ from the algorithm and $\alpha$ is the approximation ratio of the algorithm;  
	in the last step we used that $\MCA\geq \frac{\nE}{2}$. 
	This means that if $\P\neq\NP$, any polynomial time algorithm is only able to guarantee \ac{QAOA} and therefore general \ac{VQA} minimization to an optimization error $\apval\geq\frac{1-\alpha_{\max}}{2}$. 
\end{proof}
For gradient based methods this means $\apval\geq 1/4$ for logarithmically many qubits, as $\alpha = 1/2$ was shown.

\subsection{Free fermionic models}
\label{sec:VQAopt_ff}
Free fermionic models are a certain class of fermionic many-body systems that are without actual particle-particle interactions. 
They are especially interesting for us, as they can be simulated efficiently for so-called Gaussian input states and observables. 

Fermionic creation and annihilation operators are denoted by $c_j^\dagger$ and $c_j$. 
They satisfy the anticommutation relations $\{c_i^\dagger,c_j\}=\delta_{i,j}$ and $\{c_i,c_j\}=0$ for all $i,j$. 
We call an operator \emph{quadratic} or \emph{Gaussian} if it is a quadratic polynomial in the creation and annihilation operators. 
We will consider \emph{(balanced) quadratic observables} of the form 
\begin{equation}\label{eq:quadraticO}
	H = \sum_{i,j} h_{i,j}\, c^{\dagger}_{i}c_{j} 
\end{equation}
and will call $h$ the \emph{coefficient matrix} of $H$, which is Hermitian. 
Also, in the following, we denote operators by capital and their respective coefficient matrices by lowercase letters.

A quantum state is \emph{Gaussian} if it can be arbitrarily well approximated by a thermal state of a quadratic Hamiltonian. 
For a Hamiltonian $H$ we denote its ground state by
\begin{equation}\label{eq:gs}
	\rho[H]
= 
\lim_{\beta\rightarrow \infty }\frac{\e^{-\beta H}}{\Tr[\e^{-\beta H}]}\,.
\end{equation}
From this we can define the \ac{VQA} problem in the free fermionic setting.
\begin{problem}[\Ac{VQA} minimization problem, free fermions]\label{p:VQA_fermionic}\hfill\vspace{-.8\baselineskip} 
\begin{description}[noitemsep,leftmargin=0.5cm,font=\normalfont]
\item [Instance] 
Coefficient matrices $h_{\gs}, h_1, \dots, h_L, o \in \Herm(\CC^n)$. 
\item[Task]
The coefficient matrices define quadratic observables $H_\gs, H_1, \dots, H_L$ and $O$ via \eqref{eq:quadraticO} and 
$\rho_\gs= \rho[H_0]$. 
For the evolved state 
\[\rho(\vec \phi)\coloneqq U_L(\phi_L)\cdots U_1(\phi_1)\rho_{\gs}U^\dagger_1(\phi_1)\cdots U^\dagger_L(\phi_L)\,,\]
with $U_i(\phi)=\e^{-\i H_i \phi}$,
find a $\vec \phi \in \RR^L$ that minimizes $\OEx{\vec\phi}\coloneqq \Tr[O\rho(\vec\phi)]$. 

		\end{description}
\end{problem}

\begin{theorem}
	Problem~\ref{p:VQA_fermionic} is \NP-hard, even if the initial state $\rho_{\gs}$ is pure. 
\end{theorem}

\begin{proof}
	We prove the theorem via a reduction of Problem~\ref{p:shift_maxcut} to Problem~\ref{p:VQA_fermionic}. 
	Therefore, we consider a Hermitian adjacency matrix $A \in \{0,1\}^{d\times d}$. 

	For the \ac{VQA} setup, we use  $n=d\times2 $ fermionic modes $c_{i}$ with $i \in [2d]$ and $L=d$ layers. 
	To encode Problem~\ref{p:shift_maxcut} we define $h_{\gs},\{h_i\}_{i\in [L]},o\in \mathrm{Herm}(\C^{2d\times 2d})$ as follows:
	\begin{align}
	h_\gs&=\left(\1-\frac{\vec 1}{n}\right), 
	\\
	h_i&=\vec E_i\otimes \left(\begin{matrix} 1&0\\0&-1\end{matrix}\right)\ , 
	\quad i\in[d]\, ,
	\end{align}
	where $\vec 1_{a,b}=1$ and $E_{i;a,b}=\delta_{i,a,b}$ (Kronecker delta) for all $i,a,b$. 
	The coefficient matrix $o$ is given by the matrix $O$ defined in Eqs.~\eqref{eq:defOprime} and \eqref{eq:Observable_loc_VQA}, which is used for the encoding of the adjacency matrix $A$. 

	We define $\cov_{i,j}\coloneqq\Tr(c_j^\dagger c_i\rho_\gs)$ to be the correlation matrix of $\rho_\gs$, which can be evaluated to  $\cov=\vec 1/(2d)$ using the identity \eqref{eq:ap_cov}. 
	As the eigenvalues of $h_\gs$ are $\lambda=(-1,1,\cdots,1)$, $\rho_{\gs}$ describes a pure state, cp.\ Appendix~\ref{ap:therm_st}. 
	From Eq.~\eqref{eq:ap_heisev} we obtain the coefficient matrix of $O(\vec\phi)$ in the Heisenberg picture as
	\begin{align}
		o(\vec \phi)=\e^{\i h_d\phi_d}\cdots\e^{\i h_1\phi_1}o\,\e^{-\i h_1\phi_1}\cdots\e^{-\i h_d\phi_d}\,. 
	\end{align}
	With these prerequisites we can derive the following expectation value:
	\begin{equation}
	\begin{aligned}
		\OEx{\vec\phi}&=\Tr(O(\vec\phi)\rho_{\gs})\\
		&=\Tr\left(\sum_{i,j=1}^{2d}o(\vec\phi)_{i,j} c_i^\dagger c_j \rho_{\gs}\right)\\
		&=\sum_{i,j}o(\vec\phi)_{i,j} \cov_{j,i}\\
		&=\frac{1}{2d}\sum_{i,j}o(\vec\phi)_{i,j}\\
		&=\frac{1}{4}A_{i,j}\left(\cos(\phi_i)\cos(\phi_j)-1\right)=\obfun(\vec\phi)\, ,
	\end{aligned} 	
	\end{equation}
	where the last step analogously follows Eq.~\eqref{eq:VQA}. 
	As this gives the objective function from Problem~\ref{p:shift_maxcut}, this completes the desired reduction. 
\end{proof}


\section{Conclusion and outlook}
Our results show that classical training poses challenge in \ac{VQA} based hybrid quantum computations. 
Not only is optimizing \ac{VQA} algorithms \NP-hard, but also that no polynomial time algorithm can have an optimization error $\apval<1$ in all instances (assuming $\P\neq \NP$).
Additionally, for significantly simpler systems, such as those composed of logarithmically many qubits or free fermions, the hardness results already hold. 
This also shows that hardness does not merely derive from the ground state problem. We have extended these results further to optimization on a single layer of gates, to continuous unitary time evolution and to \ac{QAOA} problems. 

We encoded \NP-hard problems into local extrema of the optimization landscape of \ac{VQA} problems. 
Gradient descent type optimization and also higher order methods can converge to any local minimum, determined mostly by the initialization. From this we could explicitly show, that even for logarithmically many qubits, these methods have an approximation error of $\apval\geq \frac{1}{4}$. For our particular \ac{VQA}, this is significantly worse than what modern efficient \MaxCut\ solvers can guarantee. 
This emphasizes the need for effective initialization procedures for \ac{VQA} algorithms and poses the challenge of finding non-local heuristics for \ac{VQA} optimization to overcome the problem of these persistent local minima to reach smaller optimization errors.

In order to put our results into perspective, we briefly compare them to other hardness results for relevant optimization problems. 
For instance, optimization within the \ac{DMRG} method is \NP-hard \cite{Eisert06ComputationalDifficultyOf}. 
However, hardness holds only for errors scaling as the inverse of the bond dimension and there are variants where convergence can be rigorously guaranteed \cite{LanVazVid15}. 
\ac{VQA} optimization is arguably more similar to the optimization in the Hartree-Fock method. 
Despite being \NP-hard \cite{Schuch09ComputationalComplexityOf} it is widely used in many practical calculations. 
It is our hope that this work helps of identify and overcome optimization challenges also for practically relevant \ac{VQA} problems.

\section*{Acknowledgments}

We thank David Wierichs, Sevag Gharibian, Raphael Brieger and Thomas Wagner for helpful comments on our manuscript and Jens Watty,  Christian Gogolin and David Gross for fruitful discussions on the nature of \acp{VQE} and \acp{QAOA}. We also thank the anonymous Referee B for valuable comments, which have helped us to improve this paper.  This  work  was  supported  by  the  Deutsche Forschungsgemeinschaft   (DFG,   German   Research Foundation) via the Emmy Noether program (Grant No. 441423094) and by the German Federal Ministry of Education and Research (BMBF) within the funding program ``Quantum technologies—From basic research to  market'' in  the  joint  project  MANIQU  (Grant No. 13N15578).


\section*{Appendices}
\appendix       
\section{Proof of Lemma~\ref{l:ph_to_en} on ergodic energy spectra}
\label{ap:erd_en}
Starting from the definition of $\vec E$
\begin{equation}
E_i\coloneqq \frac{2\pi}{m^i}\, ,
\end{equation}
let $\vec \phi \in [0,2\pi)^n$ be the desired phase vector. For this we define  
\begin{equation}
s_i\coloneqq \left\lfloor\frac{\phi_i m}{2\pi}\right\rfloor \in \{0, \dots, m-1\}\, 
\end{equation}
and 
\begin{equation}
t(\vec s)\coloneqq \sum_{j=1}^{n} s_{j} m^{j-1}	\in\{0,\dots,m^n-1\}\, . 
\end{equation}
Then 
\begin{equation}
\begin{aligned}
& \phi_i - E_it
\\
&=
\left(\phi_i - \frac{2\pi s_i}{m}\right) + \left(\frac{2\pi s_i}{m} - 2\pi\sum_{j=1}^n s_{j}m^{j-1-i} \right)
\\
&=
\left(\phi_i - \frac{2\pi s_i}{m}\right) - \underbrace{\left(2\pi\sum_{j=1}^{i-1} s_{j}m^{j-1-i} \right)}_{\left|\argdot\right|\leq 2\pi/m}
\\
& \qquad \qquad \qquad\qquad
- \underbrace{\left( 2\pi\sum_{j=i+1}^n s_{j}m^{j-1-i} \right)}_{\in 2\pi \ZZ}\, .
\end{aligned}
\end{equation}
Hence, 
\begin{equation}
\modnorm{\phi_i - E_it} \leq \frac{4\pi}{m}\, ,
\end{equation}
which is what we wanted to show.
\qed

\section{Proof of Theorem~\ref{thm:QAOA2} on multilayer QAOAs}
\label{ap:multiLQAOA}
Now we construct a many-one reduction from Problem~\ref{p:shift_maxcut} to a multilayer \ac{QAOA} optimization. For this purpose, we first define some useful objects.

Let $\mc K \coloneqq \CC^d \otimes \CC^d \otimes \CC^2 \otimes \CC^2$, where $d$ will be the size of an adjacency matrix to encode {\MaxCut} and the number of layers ($L=d$). 
We define a larger Hilbert space $\H$ as a direct sum 
$\H=\H_1\oplus\dots\oplus \H_{2d+1}$ 
with $\H_\ell\cong\mc K$
for $\ell\in [2d+1]$. 
We canonically identify each $\mc H_\ell$ with the corresponding subspace $\mc H_\ell\subset \H$ and denote the canonical basis states by 
$\{\ket{i,j,a,b}_\ell\}$, where $i,j\in[d]$, $a,b\in \{0,1\}$ and $\ell$ indicates the subspace $\H_\ell$. 
Next, we define four two-level unitary evolutions by 
\begin{align*}
\ket{\psi_0(\phi)}&=\e^{-\i\phi/2}\cos(\phi/2)\ket 1+\e^{-\i\phi/2}\sin(\phi/2)\ket 2 ,
\\
\ket{\psi_1(\phi)}&=\cos(\phi)\ket 1+\sin(\phi)\ket 2 ,
\\
\ket{\psi_2(\phi)}&=\e^{\i\phi}\cos(\phi)\ket 1+\i\,\e^{\i\phi}\sin(\phi)\ket 2 ,
\\
\ket{\psi_3(\phi)}&=\e^{-\i\phi}\cos(\phi)\ket 1-\i\,\e^{-\i\phi}\sin(\phi)\ket 2 ,
\end{align*}
which are generated as $\ket{\psi_i(\phi)}=\e^{-\i H_i\phi}\ket 1$ by
\begin{align*}
H_{0}&=\frac{1}{2}\begin{pmatrix}1 &-i \\i &1\end{pmatrix},
&H_1&=\begin{pmatrix}0 &-i \\i &0\end{pmatrix},
\\
H_2&=\begin{pmatrix}-1 &-1 \\-1 &-1\end{pmatrix},&
H_3&=\begin{pmatrix}1 &1 \\1 &1\end{pmatrix}, 
\end{align*} 
with eigenvalues $\{0,1\}$, $\{-1,1\}$, $\{-2,0\}$ and $\{0,2\}$, respectively. 
Based on these evolutions, we define \emph{transfer Hamiltonians} 
$H_T,H_T^{(\kappa)}\in \Herm(\mc K\oplus\mc K)$ 
as 
\begin{equation}
\begin{split}
H_T^{(\kappa)}
\coloneqq 
\sum_{i,j,a,b,x,y} &
\begin{cases}
H_{1\,x,y}& \text{ if } i=j\quad \mathrm{or}\quad  a=0,\\
H_{2\,x,y}& \text{ if } i=\kappa\quad \mathrm{or}\quad  j=\kappa ,b=0,\\
H_{3\,x,y}& \text{ if } j=\kappa ,b=1,\\
H_{1\,x,y}& \text{ otherwise }
\end{cases}
\\
&\times\ket{i,j,a,b}_x\bra{i,j,a,b}_y
\\
\phantom . 
\\
\phantom . 
\end{split}
\end{equation}
and 
\begin{equation}
H_T\coloneqq \sum_{i,j,a,b,x,y}H_{0\, x,y}\ket{i,j,a,b}_x\bra{i,j,a,b}_y
\end{equation}
with $x,y\in \{1,2\}$. 

Let \admat\ be the adjacency matrix of an unweighted graph with at least one edge.
We will construct $H_b$ such that it has the ground state
\begin{equation}
\ket{\mathrm{gs}_b}\coloneqq 
\frac{1}{2\sqrt{\sum_{i,j} A_{i,j}}}\sum_{i\neq j,a,b}A_{i,j}\ket{i,j,a,b} \in \mc K. 
\end{equation}
For this construction it will be helpful to denote 
\begin{equation}
H_{\mathrm{gs}}\coloneqq -3\ketbra{\mathrm{gs}_b}{\mathrm{gs}_b} .
\end{equation}
The solution of \MaxCut\ will be captured by the last subspace $\H_{2d+1}\subset \H$. 
For this we define $H_p\in \Herm(\mc K)$ as
\begin{equation}
H_p=\frac{1}{2}\sum_{i,j,a,b,\tilde a,\tilde b} \delta_{a\neq \tilde a}\ket{i,j,a,b}\bra{i,j,\tilde a,\tilde b}, 
\end{equation}
where $\delta_{a \neq \tilde a} \coloneqq 1-\delta_{a,\tilde a}$. 
Finally, we define $H_b,H_c\in \Herm(\H)$ as
\begin{equation}\label{eq:HbHc}
\begin{aligned}
H_b&= H_{\mathrm{gs}}\oplus H_T^{(1)}\oplus\cdots\oplus H_T^{(d)} ,
\\
H_c&=H_T\oplus\cdots\oplus H_T\oplus H_p\, ,
\end{aligned}
\end{equation}
where the ground state of $H_b$ is $\ket{\mathrm{gs}_b}_1$ given as the embedded state $\ket{\mathrm{gs}_b}_1 = \ket{\mathrm{gs}_b}\oplus 0 \in \H$. 
Similarly, $\ket{\mathrm{gs}_b}_\ell\in \H_\ell\subset \H$ is defined.  

For the first layer, this gives the state
\begin{widetext}
\begin{equation}
	\begin{aligned}
	\ket{\Psi_0}
	& = \ket{\mathrm{gs}_b}_1=\frac{1}{2\sqrt{\sum A_{i,j}}}\sum_{i\neq j,a,b}A_{i,j}\ket{i,j,a,b}_{1} ,
	\\
	U_{c}(\gamma_1)\ket{\Psi_0}
	&=\sin(\gamma_1/2)\e^{-\i\gamma_1/2}\ket{\mathrm{gs}_b}_{2}+\e^{-\i\gamma_1/2}\cos(\gamma_1/2)\ket{\mathrm{gs}_b}_{1},
	\\
	U_b(\beta_1)U_c(\gamma_1)\ket{\Psi_0}
	&=\frac{\sin(\beta_1)\sin(\gamma_1/2)\e^{-\i\gamma_1/2}}{2\sqrt{\sum A_{i,j}}}
	\sum_{i\neq j,a,b}A_{i,j}\e^{\i a(\delta_{i,1}(\beta_1+\frac{\pi}{2})+(-1)^{b}\delta_{j,1}(\beta_1+\frac{\pi}{2}))}\ket{i,j,a,b}_{3}+\ldots;
	\end{aligned}
\end{equation}
only the highest $\H_i$ subspace is shown, as this is the relevant one. 
	Applying all $d$ layers of the \ac{QAOA} gives
\begin{equation}
	\begin{aligned}
	\ket{\Psi(\vec \beta\vec\gamma)}
	&=
	U_b(\beta_d)U_c(\gamma_d) \dots U_b(\beta_1)U_c(\gamma_1)\ket{\Psi_0}
	\\
	&=\frac{\prod_k\sin(\beta_k)\sin(\gamma_k/2)\e^{-\i\gamma_k/2}}{2\sqrt{\sum A_{i,j}}}
	\sum_{i\neq j,a,b}A_{i,j}\e^{\i a(\beta_i+\frac{\pi}{2}+(-1)^{b}(\beta_j+\frac{\pi}{2}))}\ket{i,j,a,b}_{2d+1}+\ldots
	\end{aligned}
\end{equation}
	And thus the expectation value of $H_c$ becomes
	\begin{equation}
	\begin{aligned}
	\bra{\Psi(\vec \beta\vec\gamma)}H_c\ket{\Psi(\vec \beta\vec\gamma)}
	&=\frac{1}{2}\frac{\prod_k\sin^2(\beta_k)\sin^2(\gamma_k/2)}{4\sum A_{i,j}}\\&\times
	\sum_{i,j=1}^d2A_{i,j}\left(\cos(\beta_i+\beta_j+\pi)+\cos(\beta_i-\beta_j)+\cos(-\beta_i-\beta_j-\pi)+\cos(-\beta_i+\beta_j)\right)+\braket{O_{\mathrm{rest}}}
	\\
	&=\frac{1}{2}\frac{\prod_k\sin^2(\beta_k)\sin^2(\gamma_k/2)}{4\sum A_{i,j}}\sum_{i,j=1}^d8A_{i,j}\sin(\beta_i)\sin(\beta_j) +\braket{O_{\mathrm{rest}}}
	\\
	&\geq \frac{\prod_k\sin^2(\beta_k)\sin^2(\gamma_k/2)}{\sum A_{i,j}}
	\sum_{i,j=1}^dA_{i,j}\sin(\beta_i)\sin(\beta_j)\, ,
	\end{aligned}
	\end{equation}
\end{widetext}	
where we used that $A_{i,j}^2=A_{i,j}$ and $\braket{O_{\mathrm{rest}}}\geq 0$ denotes the expectation valueof $H_c$ within
$\H_1\oplus\cdots\oplus \H_{2d}$. 
The expression 
\begin{equation*}
f(\vec\beta)\coloneqq \sum_{i,j}A_{i,j}\sin(\beta_i)\sin(\beta_j)=4\obfun(\vec \beta-\pi/2)+2\nE
\end{equation*} is minimized for a shifted solution of Problem~\ref{p:shift_maxcut} with local extrema $\beta_i \in \{\pi/2,3\pi/2\}$.
Its minimum value is non-positive, as $\MCA\geq \nE/2$. 
This means that 
\begin{equation*}
	g(\vec \beta,\vec \gamma)\coloneqq \prod_k\sin^2(\beta_k)\sin^2(\gamma_k/2)
\end{equation*}
needs to be maximized. This can be achieved trivially by setting $\gamma_i = \pi$ for all $i$ and choosing $\vec \beta$ to be a local extremum, where the function evaluates to $1$. This also minimizes $\braket{O_{\mathrm{rest}}}=0$.
This means the problem is equivalent to minimizing $\obfun(\vec \beta-\pi/2)$, which completes the reduction from Problem~\ref{p:shift_maxcut}. Similarly, it follows that an algorithm approximating this \ac{QAOA} also returns a lower bound to $\MCA$.

Finally, we show the claimed norm bounds on $H_b$ and $H_c$ from Eq.~\eqref{eq:HbHc}. 
Direct calculations reveal that $\norm{H_{\mathrm{gs}}}=3$, $\norm{H_T^{(\kappa)}} = 2$, $\norm{H_T} = 1$ and $\norm{H_p} = 1$. 
Hence, $\norm{H_c} = 1$ and $\norm{H_b} = 3$. 
\qed

\section{Free fermions} 
\label{ap:therm_st}
In this section, we provide some basics on free fermions for the special case of particle number preserving Hamiltonians. 
Throughout, we consider $n$ fermionic modes with annihilation operators $c_1, \dots, c_n$. 

First, we explain how time evolution can be simulated efficiently. 
With the commutation relation $[c_i^\dagger c_j,c_k^{\dagger}c_l]=\delta_{j,k}c_i^\dagger c_l-\delta_{i,l}c_k^\dagger c_j$ 
the time evolution in the Heisenberg picture becomes
	\begin{align}
		\dot{O}&=\i[H,O]=\i\sum_{i,j=1}^n[h,o]_{i,j}c_i^\dagger c_j \, ,
	\end{align}
where $o$ and $h$ are again the coefficient matrices of $O$ and $H$, as in \eqref{eq:quadraticO}. 	
With $\dot{O}=\sum_{i,j=1}^n \dot o_{i,j} c_i^\dagger c_j$ we obtain
		\begin{align}
		\dot{o}&=\i[h,o]\, .
		\end{align} 
For $O(t) = \e^{\i H t} O \, \e^{-\i H t}$ this gives
\begin{equation}\label{eq:ap_heisev}
	o(t) = \e^{\i h t} o \,\e^{-\i h t}\, ,
\end{equation}
meaning that the Hilbert space unitary $\e^{\i H t}$ is represented by the unitary $n\times n$ matrix $\e^{\i ht}$ on the level of second moments.

Secondly we derive an expression for the covariance matrix $\cov$ for thermal states. Quadratic observables \eqref{eq:quadraticO} can be written in a normal form. 
This form can be obtained by observing that unitary mode transformations leave the commutation relations invariant: For 
\begin{equation}
	\tilde c_i=\sum_{j=1}^n u_{i,j}c_j
\end{equation}
with $u\in \U(n)$ being unitary matrix, 
\begin{align}
\{\tilde c_i^\dagger,\tilde c_j\}
=\sum_{k,l=1}^n u_{i,k}u_{j,l}^*\{c_k ,c_l^\dagger\} 
=\delta_{i,j} \, .
\end{align}
Hence, with basic linear algebra one can find a transformation $u\in \U(n)$ such that
\begin{equation}\label{eq:nf}
	H=\sum_{i,j=1}^n h_{i,j}c_i^\dagger c_j 
=
\sum_{i,j=1}^n \tilde h_{i,j}\tilde c_i^\dagger \tilde c_j 
=
\sum_{i=1}^n \lambda_i \tilde c_i^\dagger \tilde c_i \, ,
\end{equation}
where 
$h=u^\dagger \tilde h u$ and $\tilde h = \diag(\lambda)$.
This describe $n$ decoupled modes each with eigenenergies $E_i\in\{0,\lambda_i\}$. 
The total energy of an eigenstate is therefore $E=\sum_{i=1}^n E_i$. 
We note that the ground state energy is non-degenerate if $\lambda_i\neq 0$ for all $i\in[n]$. 
The normal form \eqref{eq:nf} allows us to write the partition function of a thermal state at inverse temperature $\beta$ as 
\begin{align}\label{eq:Z}
Z=\Tr[\exp(-\beta H)]=\prod_{i}(\e^{-\beta\lambda_i}+1) \, .
\end{align}
Hence, the covariance matrix of the corresponding thermal state w.r.t.\ $\{\tilde c_i\}$ is given by
\begin{align}
\tilde\cov(\beta)_{i,j}
&= 
\bigl< \tilde c_j^\dagger \tilde c_i\bigr>_\beta
\\
&=
-\delta_{i,j}\frac{\partial}{\partial (\beta\lambda_i)}\ln(Z)
\\
&=
\frac{\delta_{i,j}}{\e^{-\beta\lambda_i}+1} \, , 
\end{align}
where we have denoted the expectation value of the thermal state by $\langle \argdot \rangle_\beta$, 
used the normal form \eqref{eq:nf} in the second step and \eqref{eq:Z} in the last step. 
In compact notation, 
\begin{equation}
	\tilde \cov(\beta) 
	= 
	(\e^{-\beta \tilde h} + 1)^{-1} \, .
\end{equation}
 
Therefore, the covariance matrix w.r.t.\ $\{c_i\}$ is
\begin{align}
	\cov(\beta)_{i,j}&=\left< c_j^\dagger c_i\right>_\beta
	=
	\sum_{k,l=1}^n u_{k,i}^*u_{l,j} 
	\left<\tilde c_l^\dagger \tilde c_k\right>_\beta\, ,
\end{align}
that is
\begin{equation}\label{eq:ap_cov}
		\cov(\beta)
	=
	u^\dagger\tilde \cov u=\left(\e^{-\beta h}+\1\right)^{-1}\, .
\end{equation}

\begin{acronym}[POVM]\itemsep.5\baselineskip
	\acro{NISQ}{noisy and intermediate scale quantum}
	\acro{VQE}{variational quantum eigensolver}
	\acro{VQA}{variational quantum algorithm}
	\acro{QAOA}{quantum approximate optimization algorithm}
	\acro{DMRG}{density matrix renormalization group}
	\acro{POVM}{positive operator valued measure}
	\acro{PVM}{projector-valued measure}
	\acro{CP}{completely positive}
	\acro{CPT}{completely positive and trace preserving}
	\acro{DFE}{direct fidelity estimation} 
	\acro{MUBs}{mutually unbiased bases} 
	\acro{SIC}{symmetric, informationally complete}
	\acro{SFE}{shadow fidelity estimation}
	\acro{RB}{randomized benchmarking}
	\acro{AGF}{average gate fidelity}
	\acro{XEB}{cross-entropy benchmarking}
	\acro{SPAM}{state preparation and measurement}
	\acro{TV}{total variation}
	\acro{HOG}{heavy outcome generation}
	\acro{BOG}{binned outcome generation}
	\acro{QPT}{quantum process tomography}
	\acro{GST}{gate set tomography}
	\acro{MW}{micro wave}
	\acro{rf}{radio frequency}
\end{acronym}

\bibliographystyle{./myapsrev4-1}
\bibliography{mk,lenbib,other} 

\end{document}